\documentclass[12pt]{article}
\usepackage{fullpage}
\usepackage{amsmath,amssymb,amsthm}
\usepackage[sort,numbers]{natbib}
\usepackage{mathrsfs}
\usepackage{todonotes}
\usepackage[max6]{authblk}
\usepackage{url}
\usepackage[ruled,vlined]{algorithm2e}
\usepackage{algpseudocode}
 
\numberwithin{equation}{section}
\newtheorem{theorem}{Theorem}
\numberwithin{theorem}{section}
\newtheorem{defprop}[theorem]{Definition-Proposition}

\newtheorem{proposition}[theorem]{Proposition}

\newtheorem{lemma}[theorem]{Lemma}
\newtheorem{corollary}[theorem]{Corollary}
\newtheorem{definition}[theorem]{Definition}

\newtheorem{conjecture}[theorem]{Conjecture}

\newtheorem{notation}[theorem]{Notation}

\newcommand{\C}{\mathbb{C}}
\newcommand{\Z}{\mathbb{Z}}
\DeclareMathOperator{\GF}{GF}

\DeclareMathOperator{\LM}{LM}
\DeclareMathOperator{\vol}{vol}

\DeclareMathOperator{\NF}{NF}
\DeclareMathOperator{\Span}{Span}
\DeclareMathOperator{\gp}{gp}

\DeclareMathOperator{\HP}{HP}

\DeclareMathOperator{\hilbert}{Hilb}
\DeclareMathOperator{\HS}{HS}
\DeclareMathOperator{\HF}{HF}
\newcommand{\sgp}{S}
\newcommand{\polytope}{\mathscr P}

\newcommand{\cone}{\mathscr C}

\DeclareMathOperator{\reg}{reg}
\DeclareMathOperator{\dwit}{d_{wit}}
\newcommand{\R}{\mathbb{R}}
\newcommand{\N}{\mathbb{N}}
\newcommand{\Q}{\mathbb{Q}}

\newcommand{\degreeI}{\delta}

\begin{document}
\title{Sparse Gr\"obner Bases: the Unmixed Case}
\author[1,2,3,4]{Jean-Charles Faug\`ere}
\author[5,6,7,8]{Pierre-Jean Spaenlehauer}
\author[2,1,3,4]{\authorcr Jules Svartz}

\affil[1] {INRIA Paris-Rocquencourt} 
\affil[2] {Sorbonne Universit\'es, UPMC Univ. Paris 06}
\affil[3] {CNRS, UMR 7606, LIP6}
\affil[4] {PolSys Project, Paris, France}
\affil[5] {INRIA Nancy Grand-Est}
\affil[6] {Universit\'e de Lorraine}
\affil[7] {CNRS,  UMR 7503, LORIA}
\affil[8] {Caramel Project, Nancy, France}

\date{Corrected version, \today}

\maketitle
\begin{abstract}
  Toric (or sparse) elimination theory is a framework developped
  during the last decades to exploit monomial structures in
  systems of Laurent polynomials. Roughly speaking, this amounts to computing in a \emph{semigroup algebra}, \emph{i.e.} an
  algebra generated by a subset of Laurent monomials.  In order to solve symbolically sparse systems, we introduce
  \emph{sparse Gr\"obner bases}, an analog of classical Gr\"obner
  bases for semigroup algebras, and we propose sparse variants of the
  $F_5$ and FGLM algorithms to compute them.  Our prototype
  ``proof-of-concept'' implementation shows large speed-ups
  (more than 100 for some examples) compared to optimized (classical)
  Gr\"obner bases software.  Moreover, in the case where the
  generating subset of monomials corresponds to the points with integer coordinates in a
  normal lattice polytope $\polytope\subset\R^n$ and under regularity
  assumptions, we prove complexity bounds which
  depend on the combinatorial properties of $\polytope$. These bounds
  yield new estimates on the complexity of solving $0$-dim systems
  where all polynomials share the same Newton polytope (\emph{unmixed case}).
  For instance, we generalize the bound $\min(n_1,n_2)+1$ on
  the maximal degree in a Gr\"obner basis of a $0$-dim. bilinear
  system with blocks of variables of sizes $(n_1,n_2)$ to the
  multilinear case: $\sum n_i - \max(n_i)+1$. We also propose a
  variant of Fr\"oberg's conjecture which allows us to estimate the
  complexity of solving overdetermined sparse systems.
\end{abstract}

\section{Introduction} {\bf Context and problem statement.} Many
polynomial systems or systems of Laurent polynomials arising in
applications do not have a dense monomial structure (\emph{e.g}
multi-homogeneous systems, fewnomials, systems invariant under the action of a linear
group,\ldots). The development of
toric geometry during the 70s/80s has led to toric (or
sparse) elimination theory \cite{Stu91}, a framework designed to
study and exploit algorithmically these monomial structures.

Central objects in toric geometry are \emph{semigroup algebras} (also
called toric rings).  If $\sgp\subset \Z^n$ is an affine semigroup
(see Def.~\ref{def:affinesgp}), then the semigroup algebra $k[\sgp]$
is the set of finite sums $\sum_{s\in\sgp} a_s X^s$, where $X$ is a
formal symbol, $k$ is a field, $a_s\in k$ and $s\in\sgp$. Semigroup
algebras are isomorphic to sub\-algebras of $k[X_1^{\pm 1},\ldots, X_n^{\pm 1}]$
generated by a finite subset of monomials.

Our motivation is to propose fast algorithms to solve
symbolically systems whose support lie in one of the following classes
of semigroups: semigroups constructed from the points with integer coordinates
in a normal lattice polytope $\polytope\subset\R^n$ (in that case, the
algorithms we propose are well-suited for \emph{unmixed} systems: the
Newton polytopes of the input polynomials are all equal to $\polytope$) or
semigroups generated by a scattered set of monomials (fewnomial
systems).

{\bf Main results.}  
Given a $0$-dim. system of Laurent polynomials $f_1=\dots=f_m=0$ and a
finite subset $M\subset\Z^n$ such that each polynomial
belongs to the subalgebra generated by $\{X_1^{\alpha_1}\cdots
X_n^{\alpha_n}\mid \alpha\in M\}$, we associate to $M$ two affine
semigroups: $\sgp_M\subset\Z^n$ generated by $M$ and
$\sgp_M^{(h)}\subset\Z^{n+1}$ generated by
$\{(\alpha,1)\in\Z^{n+1}\mid \alpha\in M\}$. Under the assumption that
$\sgp_M$ contains zero but no nonzero pairs $(s_1,s_2)\in S_M^2$
s.t. $s_1+s_2=\mathbf 0$, our solving strategy proceeds by combining a
sparse variant in the homogeneous algebra $k[\sgp_M^{(h)}]$ of the
{\tt MatrixF5} algorithm and a sparse variant in $k[\sgp_M]$ of the {\tt FGLM}
algorithm. We define a notion of \emph{sparse Gr\"obner basis}
(Def.~\ref{def:sGB}) that is computed by the {\tt sparse-MatrixF5} algorithm
if we know a bound on its maximal degree (this maximal degree is
called the \emph{witness degree} of the system).
An important feature of sparse GBs is that their definition depends
only on the ambient semigroup algebra and not on an embedding in a polynomial
algebra.
In this sense, they differ conceptually from SAGBI bases, even though
the {\tt sparse-FGLM} algorithm has similarities with the {\tt
  SAGBI-FGLM} algorithm proposed in \cite{FR09}. In the special case $S_M=\N^n$,
sparse Gr\"obner bases in $k[S_M]$ are classical Gr\"obner bases, and {\tt sparse-FGLM} is the usual {\tt FGLM} algorithm.

At the end of the solving
process, we obtain a rational parametrisation of the form
\[Q(T)=0 \quad \text{and} \quad \forall \alpha\in M\setminus\{\mathbf 0\},\quad X_1^{\alpha_1}\cdots X_n^{\alpha_n}-Q_{\alpha}(T)=0\]
where $Q\in k[T]$ is a univariate polynomial, and for all $\alpha\in M$, $Q_{\alpha}\in k(T)$ is a rational function.
Consequently, the solutions of the input sparse system 
can be expressed in terms of the roots of the univariate polynomial $Q$ by inverting a monomial map.

The next main result addresses the question of the complexity of this solving process when $M$ is given as the set $\polytope\cap\Z^n$, where $\polytope\subset\R^n$ is a lattice polytope of dimension $n$. It turns out that the complexities of {\tt sparse-MatrixF5} and {\tt sparse-FGLM} algorithms depend mainly on intrinsic combinatorial properties of $\polytope$:
\begin{itemize}
\item the normalized volume $\vol(\polytope)\in\N$;
\item the Castelnuovo-Mumford regularity $\reg(k[\sgp_{\polytope\cap\Z^n}^{(h)}])=n+1-\ell$ where $\ell$ is the smallest integer such that the intersection of $\Z^n$ with the interior of $\ell\cdot\polytope$ is nonempty;
\item the Ehrhart polynomial $\HP_\polytope(\ell)$ which equals the cardinality of $(\ell\cdot\polytope)\cap \Z^n$ for $\ell\in\N$.
\end{itemize}

We use as indicator of the complexity the \emph{witness degree} which
bounds the maximal ``sparse degree'' (corresponding to an $\N$-grading
on $k[\sgp_{\polytope\cap\Z^n}^{(h)}]$) in a reduced sparse Gr\"obner
basis.  More precisely, we obtain the following complexity estimates:

\begin{theorem}\label{theo:mainres}
Let $\polytope\subset\R^n$ be a normal lattice polytope of dimension
$n$ with one vertex at $\mathbf 0\in\Z^n$, $(d_1,\ldots,d_n)$ be a sequence of positive integers and
$(f_1,\ldots, f_n)$ be a regular sequence of Laurent polynomials in $k[X_1^{\pm 1},\ldots, X_n^{\pm 1}]^n$, such that the support of $f_i$
is included in $\{X_1^{s_1}\cdots X_n^{s_n}\mid s\in (d_i\cdot\polytope)\cap \Z^n\}$. Then a sparse GB of the
ideal $\langle f_1,\ldots, f_n\rangle\subset k[\sgp_{\polytope\cap \Z^n}]$ can be computed
within
$$O\left(n \HP_\polytope(\dwit)^\omega\right)$$ arithmetic operations
in $k$, where 
$\omega<2.373$ is a feasible exponent for the matrix multiplication
and $\dwit\leq\reg(k[\polytope])+1+\sum_{j=1}^n(d_j-1)$.
Moreover, if $\mathbf 0$ is a simple vertex of $\polytope$
(\emph{i.e.} a vertex which is the intersection of $n$ facets), then the
{\tt sparse-FGLM} algorithm executes at most
$$O\left(\HP_\polytope(1) \left(\vol(\polytope) \prod_{j=1}^n d_j\right)^3\right)$$
arithmetic operations in $k$.
\end{theorem}

Direct consequences of these formulas allow us to derive new
complexity bounds for solving regular multi-homo\-ge\-ne\-ous systems.
We show that the witness degree of a regular system of
$n$ multi-homogeneous polynomials of multi-degree $(d_1,\ldots,d_p)$
w.r.t. blocks of variables of sizes $(n_1,\ldots, n_p)$ (with $\sum
n_i=n$) is bounded by $n+2-\max_{i\in\{1,\ldots,p\}}(\lceil
(n_i+1)/d_i\rceil)$ (which generalizes the bound $\min(n_1,n_2)+1$ in
the bilinear case \cite{faugere2011grobner}). We also propose a variant of Fr\"oberg's
conjecture for sparse systems and a notion of semi-regularity, which
yield complexity estimates for solving sparse overdetermined systems.

We have implemented in C a prototype of the {\tt sparse-MatrixF5}
algorithm, that runs several times faster than the original $F_5$
algorithm in the FGb software. For instance, we report speed-up ratios
greater than 100 for instances of overdetermined bihomogeneous
systems. The implementation also works well for fewnomial systems
(although this case is not covered by our complexity analysis).

 {\bf Related works.}  Computational aspects of toric geometry
and Gr\"obner bases are investigated in \cite{Stu96}. In particular,
\cite[Subroutine 11.18]{Stu96} gives an algorithm to compute syzygies
of monomials in toric rings, which is an important routine for critical-pairs based algorithms.

Other approaches have been designed to take advantage of the sparse
structure in Gr\"obner bases computations. For instance, the Slim Gr\"obner
bases in \cite{brickenstein2010slimgb} describes strategies to avoid
increasing the number of monomials during computations. This
approach improves practical computations, but does not lead to new
asymptotic complexity bounds for classes of sparse systems.

The sparse structure and the connection with toric geometry have also
been incorporated to the theory of resultants, and a vast literature
has been written on this topic, see \emph{e.g.}
\cite{emiris2002symbolic, emiris2005toric,
  Canny99asubdivision,canny1993efficient}. In particular, mixed
monomials structures are well-understood in this context. Although we
do not know how to extend the algorithms proposed in this paper to
mixed structures, Gr\"obner-type algorithms enjoy the property of
extending without any modification to the overdetermined case.

 {\bf Perspectives.} Our approach is for the moment limited
to \emph{unmixed systems}: all input polynomials have to lie in the
same semigroup algebra. A possible extension of this work would be the
generalization to mixed systems (where the algorithms would depend on
the Newton polytope of each of the polynomials of the system). Some
results seem to indicate that such a generalization may be possible:
for instance, under genericity assumptions, mixed monomial bases of
quotient algebras are explicitly described in
\cite{pedersen1995mixed}. Also, a bound on the witness degree and the
complexity analysis is for the moment restricted to the polytopal
case. Merging the approach in this paper with a Buchberger's type
approach such as \cite[Algo.~11.17]{Stu96} could lead to a termination
criterion of the {\tt sparse-MatrixF5} algorithm in the non-regular cases
and for positive dimensional systems. Finally, finding complexity
bounds which explain the efficiency of the sparse Gr\"obner bases
approach for fewnomial systems (see Table \ref{table:fewnomials})
remains an open problem.

 {\bf Organisation of the paper.}  We recall in Section \ref{sec:background}
the background material on semigroup algebras and convex geometry that
will be used throughout this paper. Section \ref{sec:sGB} introduces sparse Gr\"obner bases and describes a
general solving process for sparse systems. The main algorithms
are described in Section \ref{sec:algos} and their complexities are analyzed in
Section \ref{sec:complexity}. Finally, we describe in Section \ref{sec:applis} some results
that are direct consequences of this new framework and experimental results in Section \ref{sec:expe}.  

 {\bf Acknowledgements.} We are grateful to Kaie Kubjas,
Guillaume Moroz and Bernd Sturmfels for helpful discussions and for
pointing out important references. This work was partly done while the
second author was supported and hosted by the Max Planck Institute for
Mathematics (Bonn, Germany). This work was partly supported by the HPAC grant of the French National Research Agency (HPAC ANR-11-BS02-013).

\section{Preliminaries and notations}\label{sec:background}
In this paper, the basic algebraic objects corresponding to monomials
in classical polynomial rings are \emph{affine semigroups}. We always
consider them embedded in $\Z^n$. We refer the reader to
\cite{MilStu05,CoxLitSch11,fulton1993introduction} for a more detailed
presentation of this background material.  First, we describe the main
notations that will be used throughout the paper:
\begin{definition}\label{def:affinesgp}
  An \emph{affine semigroup} $\sgp$ is a finitely-gene\-ra\-ted additive subsemigroup of
  $\Z^n$ for some $n\in\N$ containing $\mathbf 0\in\Z^n$ and no
  nonzero invertible element (\emph{i.e.} for all $s,s'\in\sgp\setminus\{\mathbf 0\}, s+s'\ne
  \mathbf 0$). Any affine semigroup has a unique minimal set of
  generators, called the \emph{Hilbert basis} of $\sgp$ and denoted by
  $\hilbert(\sgp)$. Let $\gp(\sgp)$ denote the smallest subgroup of
  $\Z^n$ containing $\sgp$. Then $\sgp$ is called \emph{normal} if
  $\sgp=\{q\in\gp(\sgp)\mid\exists c\in\N, c\cdot q\in \sgp\}$. For a
  field $k$, we let $k[\sgp]$ denote the associated \emph{semigroup
    algebra} of finite formal sums $\sum_{s\in \sgp}a_s X^s$
  where $a_s\in k$. An element $X^s\in k[\sgp]$ is called a
  \emph{monomial}.

  We use the letter $M$ to denote a finite subset of $\Z^n$ such that
  $\mathbf 0\in M$ and the semigroup $\sgp_M$ generated by $M$ contains no nonzero
  invertible element. Also, we let $\sgp_M^{(h)}$ denote the affine
  semigroup generated by $\{(\alpha,1)\mid \alpha\in
  M\}\subset\Z^{n+1}$. The semigroup algebra $k[\sgp_M^{(h)}]$ is
  \emph{homogeneous} (\emph{i.e.} $\N$-graded and generated by degree
  $1$ elements): the degree of a monomial $X^{(s_1,\ldots,s_n,d)}$ is
  $d\in\N$. The vector space of homogeneous elements of degree $d\in\N$ in $k[\sgp_M^{(h)}]$ is denoted by $k[\sgp_M^{(h)}]_d$.
\end{definition}

Depending on the articles on this topic, the condition ``$\sgp$
contains no invertible element'' is not always included in the
definition of an affine semigroup. However, this is a necessary
condition for the algorithms we propose in this paper. Also, the term
``Hilbert basis'' is sometimes reserved for affine semigroups of the
form $\cone\cap\Z^n$ where $\cone$ is a rational cone (see \emph{e.g.}
\cite[Prop.~7.15]{MilStu05} and the discussion after this
statement). We always assume implicitely that $\gp(S)\subset\Z^n$ is a
full rank lattice (this does not lose any generality since this case
can be reached by embedding $S$ in a lower dimensional
$\Z^{n'}$). Note that $k[\N^n]$ is the classical polynomial ring
$k[X_1,\ldots, X_n]$. Semigroup algebras are integral domains
\cite[Thm. 7.4]{MilStu05} of Krull dimension $n$ and play an important
role in toric geometry: they are precisely the coordinate rings of
\emph{affine toric varieties}.  

The normality of the semigroup $\sgp$
is an important property which implies that $k[\sgp]$ is
Cohen-Macaulay by a theorem by Hochster \cite{Hoc72}.
An important feature of normal affine semigroups is that they can be represented by the intersection of $\Z^n$ with a pointed rational polyhedral cone (also called \emph{stron\-gly convex rational polyhedral cone} \cite[Sec 1.1]{Oda88}).
\begin{definition}
A \emph{cone} $\cone\subset \R^n$ is a convex subset of $\R^n$ stable by multiplication by $\R_+$, the set of non-negative real numbers. The dimension $\dim(\cone)$ of a cone $\cone$ is the dimension of the linear subspace spanned by $\cone$. A cone is called \emph{pointed} if it does not contain any line. A pointed cone of dimension $1$ is called a ray. A ray is called \emph{rational} if it contains a point in $\Z^n$. A \emph{rational polyhedral cone} is the convex hull of a finite number of rational rays. Pointed rational polyhedral cones will be abbreviated \emph{PRPC}. 
\end{definition}

We shall use PRPCs in Section \ref{sec:sGB} to define admissible
monomial orderings in semigroup algebras.
We now recall the definition of \emph{simplicial affine semigroups}, for which we will be able to derive tight complexity bounds for the {\tt sparse-FGLM} algorithm (Section~\ref{sec:complexity}).

\begin{definition}\label{def:simplicial}
  An affine semigroup $S\subset \Z^n$ is called simplicial if the convex hull of $\R_{+}S$ is a simplicial PRPC, \emph{i.e.} the convex hull of $n$ linearly independant rays. 
\end{definition}

Another important family of objects are \emph{projective toric
  varieties}. Their homogeneous coordinate rings are associated to a
lattice polytope, which we shall assume to be normal in order to ensure that the coordinate ring is Cohen-Macaulay. As in the classical case, homogeneity is a
central concept to analyze the complexity of Gr\"obner bases
algorithms. All lattice polytopes will be assumed to be full dimensional.

\begin{definition}
A \emph{lattice polytope} $\polytope\subset\R^n$ is the convex hull of a
finite number of points in $\Z^n$. Its \emph{normalized volume}, \emph{i.e.} $n!$ times its Euclidean volume, is
denoted by $\vol(\polytope)\in\N$. 
To a lattice polytope $\polytope\subset\R^n$ is associated an affine
semigroup $S_{\polytope\cap\Z^n}^{(h)}\subset\Z^{n+1}$ generated by $\{(\alpha,1)\mid
\alpha\in\polytope\cap\Z^n\}$. The polytope $\polytope$ is called
\emph{normal} if $\sgp_{\polytope\cap\Z^n}^{(h)}$ is a normal semigroup. The
associated semigroup algebra is called a \emph{polytopal algebra} and
abbreviated $k[\polytope]$.
\end{definition}
If $\polytope\subset\R^n$ is a lattice polytope containing $\mathbf 0$ as a vertex, then $k[\polytope]=k[S_{\polytope\cap\Z^n}^{(h)}]$ (Def.~\ref{def:affinesgp}). Moreover, if $\polytope$ is normal, then so is $\sgp_{\polytope\cap \Z^n}$ \cite[Prop. 2.17]{CoxLitSch11}.
Also, note that if $\polytope'$ is a translation of $\polytope$, then the homogeneous algebras $k[\polytope]$ and $k[\polytope']$ are isomorphic. Consequently, we shall assume w.l.o.g. in the sequel that one of the vertices of $\polytope$ is the origin, so that $M=\polytope\cap\Z^n$ verifies the assumptions of Def.~\ref{def:affinesgp}. We also introduce a few more notations for lattice polytopes:

\begin{notation}
The number of lattice points in a polytope $\polytope\subset\R^n$
(\emph{i.e.} the cardinality of $\polytope\cap \Z^n$) is denoted by
$\#\polytope$.  The \emph{Minkowsky sum} of two lattice polytopes
$\polytope_1,\polytope_2\subset\R^n$ is the lattice polytope
$\{p_1+p_2 \mid p_1\in \polytope_1,p_2\in \polytope_2\}$. For all
$\ell\in \N$ we write $\ell\cdot\polytope$ for the Minkowski sum
$\polytope + \dots +\polytope$ with $\ell$ summands.  For $n\in\N$, we
let $\Delta_n\subset\R^n$ denote the \emph{standard simplex}, namely
the convex hull of $\mathbf 0$ and of the points $\mathbf e_i\in\R^n$
whose entries are zero except for the $i$th coefficient which is equal
to $1$.  For $\polytope_1\subset\R^i$,$\polytope_2\subset\R^j$ we write
$\polytope_1\times\polytope_2 \subset\R^{i+j}$ for the lattice polytope
whose points are $\{(p_1,p_2) \mid p_1\in\polytope_1,
p_2\in\polytope_2\}.$ 
\end{notation}

Next, we recall several useful classical properties of polytopal algebras. We refer to \cite[Ch.~12]{MilStu05} for a detailed presentation of the connections between Ehrhart theory and computational commutative algebra.
\begin{proposition}\label{prop:HSpolytope}
Let $\polytope\subset \R^n$ be a lattice polytope. For $d\in\N$, we
let $\HP_\polytope\in\Q[d]$ denote the \emph{Ehrhart polynomial} of
$\polytope$, \emph{i.e.} $\HP_\polytope(d)=\#(d\cdot\polytope)$. Also, let $\HS_\polytope(t)\in\Z[[t]]$ denote the generating series
$$\HS_\polytope(t)=\sum_{d\in \N}\HP_\polytope(d) t^d.$$
Then the Hilbert series of the polytopal algebra $k[\polytope]$, namely
$$\HS_{k[\polytope]}(t)=\sum_{d\in\N}\dim_k (k[\polytope]_d)t^d$$
is equal to $\HS_\polytope$ and there exists a polynomial $Q\in \Z[t]$ with non-negative coefficients such that
$$\HS_{\polytope}(t)=\dfrac{Q(t)}{(1-t)^{n+1}},\quad\deg(Q)\leq n.$$
\end{proposition}

\begin{proof}
The fact that the map $\HP_\polytope: d\mapsto \#(d\cdot\polytope)$ is polynomial is a classical result by Ehrhart \cite{ehrhart1962polyedres}. The second statement $\HS_\polytope=\HS_{k[\polytope]}$ follows from the definition of $k[\polytope]$. The last statement is Stanley's non-negativity theorem \cite[Thm.~2.1]{stanley1980decompositions}.
\end{proof}

We let $\reg(k[\polytope])$ denote the \emph{Castelnuovo-Mumford regularity} of $k[\polytope]$. The Castenuovo-Mumford regularity of a graded module is an important measure of its ``complexity'': it is related to the degrees where its local cohomology modules vanish. We refer to \cite[Ch.~15]{brodmann2012local} for a more detailed presentation. The following classical proposition relates the regularity with a combinatorial property of the polytope $\polytope$ and with the degree of the numerator of $\HS_\polytope$:

\begin{proposition}\label{prop:regpolytope}
Let $\polytope$ be a normal lattice polytope.  The regularity
$\reg(k[\polytope])$ equals $n-\ell+1$, where $\ell$ is the
smallest integer such that $\ell\cdot \polytope$ contains an integer
point in its interior. Moreover, with the same notations as
in Proposition \ref{prop:HSpolytope}, $\deg(Q)=\reg(k[\polytope])$.
\end{proposition}
\begin{proof}
The first claim follows from \cite[Sec.~5.4]{bruns1997normal}. To prove the second claim, we use the partial fraction expansion of $\HS_\polytope$ which is of the form $\sum_{\ell=n+1-\deg(Q)}^{n+1}\frac{a_\ell}{(1-t)^\ell}$ with $a_{n+1-\deg(Q)}\neq 0$. Then we obtain the equality $\HP_\polytope(d)=\sum_{\ell=n+1-\deg(Q)}^{n+1}\frac{a_\ell}{(\ell-1)!}\prod_{j=1}^{\ell-1}(d+j)$, and hence $d=n-\deg(Q)+1$ is the smallest positive integer such that $\HP_\polytope(-d)\neq 0$.
The Ehrhart-MacDonald reciprocity \cite{Mac71} concludes the proof.
\end{proof}

\section{Sparse Gr\"obner bases}\label{sec:sGB}
In this section, we show that classical Gr\"obner bases algorithms
extend to the context of semigroup algebras.  First, we need to extend
the notion of admissible monomial ordering and of Gr\"obner bases.  We
recall that the monomials of a semigroup algebra $k[\sgp]$ are the
elements $X^s$ for $s\in\sgp$.

\begin{definition}\label{def:sGB}
Let $\sgp$ be an affine semigroup. A total ordering on the monomials of $k[\sgp]$ is called \emph{admissible} if 
\begin{itemize}
\item it is compatible with the internal law of $\sgp$: for any $s_1,s_2,s_3\in \sgp$, $X^{s_1}\prec X^{s_2}\Rightarrow X^{s_1+s_3}\prec X^{s_2+s_3}$;
\item for any $s\in \sgp\setminus\{\mathbf 0\}$, $X^{\mathbf 0}\prec X^{s}$.
\end{itemize}
For a fixed admissible ordering $\prec$ and for any element $f\in k[\sgp]$, we let $\LM(f)$ denote its leading monomial. Similarly, for any ideal $I\subset k[\sgp]$, $\LM(I)$ denotes the ideal generated by $\{\LM(f) \mid f \in I\}$. A finite subset $G\subset I$ is called a \emph{sparse Gr\"obner basis} (abbreviated sGB) of $I$ with respect to $\prec$ if the set $\{\LM(g) \mid g \in G\}$ generates $\LM(I)$ in $k[\sgp]$.
\end{definition}
Note that admissible orderings exist for any semigroup algebra: the
convex hull of a semigroup $\sgp\subset\Z^n$ is a PRPC $\cone\subset\R^n$
(this is a consequence of the fact that there is no nonconstant invertible monomial
 in $k[\sgp]$).  Now one can pick $n$ independant
linear forms $(\ell_1,\ldots,\ell_n)$ with integer coefficients in the
dual cone $\cone^*=\{\text{linear forms }\ell:\R^n\rightarrow \R \mid
\forall \mathbf x\in \cone,\ell(\mathbf x)\geq 0\}$, and set $X^{s_1}\prec X^{s_2}$ if
and only if the vector $(\ell_1(s_1),\ldots,\ell_n(s_1))$ is smaller
than $(\ell_1(s_2),\ldots,\ell_n(s_2))$ for a classical admissible
ordering on $\N^n$.

Note that the assumption that $k[\sgp]$ contains no nonconstant invertible monomial is
a necessary and sufficient condition for the existence of an admissible ordering.

\smallskip

We describe now an algorithmic framework to solve sparse systems of Laurent polynomials.
Let $M\subset \Z^n$ be a finite subset verifying the assumptions of Definition \ref{def:affinesgp}, and $f_1,\ldots, f_m\in k[X_1^{\pm 1},\ldots, X_n^{\pm 1}]$ be Laurent polynomials such that the supports of the $f_i$ are included in $\{X_1^{\alpha_1}\cdots X_n^{\alpha_n}\mid\alpha\in S_M\}$. 
Note that translating $M$ amounts to multiplying the Laurent polynomials by Laurent monomials: this does not change the set of solutions of the system in the torus $\left(\overline k\setminus\{0\}\right)^n$.

Assuming that the system $f_1=\dots=f_m=0$ has finitely-many solutions in $(\overline k\setminus\{0\})^n$, we proceed as follows:
\begin{enumerate}
\item homogenize $(f_1,\ldots, f_m)$ via Def.-Prop.~\ref{def:hom} (note that the homogenization depends on the choice of the (not necessarily minimal) generating set $M$;
\item compute a sparse Gr\"obner
  basis w.r.t. a graded ordering of the homogeneous ideal $I=\langle f_1^{(h)},\ldots, f_m^{(h)}\rangle\subset
  k[\sgp_M^{(h)}]$ by using a variant of $F_4$/$F_5$
  algorithm (Algo.~\ref{algo:matrixF5}). 
\item dehomogenize the output to obtain a sGB of the ideal $\langle f_1,\ldots,f_m\rangle\subset k[\sgp_M]$ (Prop.~\ref{prop:deshom});
\item use a sparse variant of FGLM to obtain a $0$-dim. triangular system (hence containing a univariate polynomial) whose solutions are the image of the toric solutions of $f_1=\dots=f_m=0$ by monomial maps (Algo.~\ref{algo:sparse-fglm});
\item compute the non-zero roots of the univariate polynomial and invert the monomial map to get the solutions.
\end{enumerate}

We focus on the four first steps of this process.  The fifth step
involves computing the roots of a univariate polynomial, for which
dedicated techniques exist and depend on the field $k$. It also
involves inverting a monomial map, which can be achieved by solving a
consistent linear system of $\#\hilbert(S_M)$ equations in $n$ unknowns.

In the sequel of this section, we investigate the behavior of sparse
Gr\"obner bases under homogenization and dehomogeneization (Steps 1 and
3). We refer the reader to \cite[Ch. 2]{CoxLitSch11} for geometrical
aspects of projective toric varieties and their affine charts. If $M$
verifies the assumptions of Def.~\ref{def:affinesgp}, then there is a
canonical dehomogenization map:

\begin{definition}\label{def:dehom}
With the notations of Def.~\ref{def:affinesgp}, there is a \emph{dehomogeneization morphism} $\chi_M$ defined by
$$\begin{array}{rccc}
\chi_M:&k[\sgp_M^{(h)}]&\rightarrow&k[\sgp_M]\\
&X^{(s,d)}&\mapsto&X^{s}
\end{array}$$
\end{definition}

\begin{defprop}\label{def:hom}
With the notations of Def. \ref{def:affinesgp}, for any $f\in k[S_M]$, we call degree of $f$, the number $\deg(f)=\min \{d\in\N\mid \chi_M^{-1}(f)\cap k[S_M^{(h)}]_d\neq \emptyset\}$. Moreover the set $\chi_M^{-1}(f)\cap k[S_M^{(h)}]_{\deg(f)}$ contains a unique element, called the \emph{homogenization} of $f$.
\end{defprop}
\begin{proof}
The only statement to prove is that $\chi_M^{-1}(f)\cap k[S_M^{(h)}]_{\deg(f)}$ contains a unique element. Let $f_1^{(h)},f_2^{(h)}\in \chi_M^{-1}(f)\cap k[S_M^{(h)}]_{\deg(f)}$. Then $\chi_M(f_1^{(h)}-f_2^{(h)})=0$, which implies $f_1^{(h)}=f_2^{(h)}$.
\end{proof}

The next step is to prove that dehomogenizing a homogeneous Gr\"obner basis (with respect to a graded ordering) gives a Gr\"obner basis of the dehomogenized ideal.

\begin{definition}
An admissible monomial ordering $\prec$ on $k[\sgp_M^{(h)}]$
 is called \emph{graded} if there exists an associated ordering $\prec'$ on $k[\sgp_M]$ such that
$$
X^{(s_1,d_1)}\prec X^{(s_2,d_2)} \Leftrightarrow \left\{\begin{array}{c}d_1<d_2\text{ or}\\
d_1=d_2\text{ and } X^{s_1}\prec'X^{s_2}
 \end{array}\right.$$
\end{definition}
\begin{proposition}\label{prop:deshom}
Let $G$ be an homogeneous sGB of an homogeneous ideal $I\subset k[\sgp_M^{(h)}]$ with respect to a graded ordering. Then 
$\chi_M(G)$ is a sGB of $\chi_M(I)$ with respect to the associated ordering on $k[\sgp_M]$.
\end{proposition}

\begin{proof}
First, notice that $\chi_M$ commutes with leading monomials on homogeneous components of $k[\sgp_M^{(h)}]$: for any $f\in k[\sgp_M^{(h)}]_d$, $\chi_M(\LM(f))=\LM(\chi_M(f))$. Let $f\in\chi_M(I)$ and $f^{(h)}\in I$ be a homogeneous polynomial such that $f$ is equal to $\chi_M(f^{(h)})$. Consequently, there exists $g\in G$ such that $\LM(g)$ divides $\LM(f^{(h)})$. Applying $\chi_M$, we obtain that $\LM(\chi_M(g))$ divides $\LM(\chi_M(f^{(h)}))=\LM(f)$. Therefore $\chi_M(G)$ is a sGB of $\chi_M(I)$ for the associated ordering.
\end{proof}

\section{Algorithms}\label{sec:algos}
\subsection{Sparse-MatrixF5 algorithm}
As pointed out in \cite{lazard1983grobner}, classical Gr\"obner bases algorithms are related to linear algebra via the Macaulay matrices.
Since $k[\sgp_M^{(h)}]$ is generated by elements of degree $1$, the following proposition shows that similar matrices can be constructed in the case of semigroup algebras:
\begin{proposition}
Any monomial of degree $d$ in $k[\sgp_M^{(h)}]$ is equal to a product
of a monomial of degree $d-1$ by a monomial of degree $1$.
\end{proposition}

With the notations of Def.~\ref{def:affinesgp}, $k[\sgp_M^{(h)}]$ has the following property: for any $f_1\in k[\sgp_M^{(h)}]_d$, and for all $\ell\geq d$, there exists $f_2\in k[\sgp_M^{(h)}]_\ell$ s.t. $\chi_M(f_1)=\chi_M(f_2)$. This leads to the following definition of a $D$-Gr\"obner basis:
\begin{definition}
  Let $I\subset k[\sgp_M^{(h)}]$ be a homogeneous ideal and $\prec$ be an admissible monomial ordering on $k[\sgp_M]$. Then a finite subset $G\subset I$ is called a $D$-sGB of $I$ if for any homogeneous polynomial $f\in I$ with $\deg(f)\leq D$, there exists $g\in G$ such that $\LM(g)$ divides $\LM(f)$.
\end{definition}
 Note that for any $D\in\N$ there always exists a homogeneous $D$-sGB of $I$. A $D$-sGB of $I$ can be deduced from a row echelon basis of the $k$-vector space $I\cap k[\sgp_M^{(h)}]_D$, and can be computed via the Macaulay matrix:
\begin{definition}
  Let $f_1,\ldots,f_m$ be homogeneous polynomials in $k[\sgp_M^{(h)}]$. Then the \emph{Macaulay matrix} in degree $d\in\N$ of $f_1,\ldots, f_m$ is a matrix with $\sum_{i=1}^m\max(\HF_{k[\sgp_M^{(h)}]}(d-\deg(f_i)),0)$ rows, $\HF_{k[\sgp_M^{(h)}]}(d)$ columns and entries in $k$ (where $\HF_{k[S_M^{(h)}]}$ is the Hilbert function of $k[S_M^{(h)}]$). Rows are indexed by the products $X^{(s,d-\deg(f_i))}\cdot f_i$ where $X^{(s,d-\deg(f_i))}\in k[\sgp_M^{(h)}]$. Columns are indexed by monomials of degree $d$ and are sorted in decreasing order w.r.t. an admissible monomial ordering. The entry at the intersection of the row $X^{(s,d-\deg(f_i))}\cdot f_i$ and the column $X^{(s',d)}$ is the coefficient of $X^{(s',d)}$ in $X^{(s,d-\deg(f_i))}\cdot f_i$.
\end{definition}

By a slight abuse of notation, we identify implicitely a row in the Macaulay matrix of degree $d$ with the corresponding polynomial in $k[\sgp_M^{(h)}]_d$. The relation between the Macaulay matrix and a $D$-sGB is given by:
\begin{defprop}\label{defprop:dwit}
Let $f_1,\ldots, f_m\in k[S_M^{(h)}]$ be homogeneous polynomials, $\prec$ a graded monomial ordering, and for $d\in\N$, let $G_d$ be the set of polynomials corresponding to the rows of the reduced row echelon form of the Macaulay matrix in degree $d$ of $f_1,\ldots f_m$. Then we have
$$\begin{array}{c}
\text{for any $D\in\N$ }, G_0\cup \dots\cup G_D\text{ is a $D$-sGB of $I$},\\
\text{and }\chi_M(G_0)\subset\chi_M(G_1)\subset\chi_M(G_2)\subset \dots
\end{array}$$
The smallest integer $\ell$ such that $\chi_M(G_\ell)$ is a sGB of the ideal $\chi_M(\langle f_1,\ldots, f_m\rangle)$ is called \emph{the witness degree} and noted $\dwit$. 
\end{defprop}

\begin{proof}
The first statement ($G_0\cup \dots\cup G_D\text{ is a $D$-sGB of
  $I$}$) follows from the fact that $G_d$ is a triangular basis of the
vector space $k[S_M^{(h)}]_d$. The second statement is deduced from
the inclusions $\chi_M(k[S_M^{(h)}]_0)\subset
\chi_M(k[S_M^{(h)}]_1)\subset\dots$. Let $G$ be a sGB of $\langle f_1,\ldots, f_m\rangle$. Then $\dwit$ is bounded above by $\max\{\deg(g)\mid g\in G\}$ and is therefore finite.
\end{proof}

As in the original $F_5$ algorithm~\cite{Fau02a}, many lines are reduced to 0 during row-echelon form computations of Ma\-cau\-lay matrices. The $F_5$ criterion~\cite{Fau02a,EF14}\cite[Prop.~6]{BFS13} extends without any major difficulty in this context and identifies all reductions to zero when the input system is a regular sequence in $k[\sgp_M^{(h)}]$:
\begin{lemma}[$F_5$-criterion]\label{lem:F5crit}
  With the notations of Algorithm \ref{algo:matrixF5}, 
  if \(m\) is the leading monomial of a row in $\widetilde{\mathcal{M}}_{d-d_i,i-1}$ then the polynomial $m f_i$ belongs to the vector space \[
  \Span_k(Rows(\mathcal{M}_{d,i-1}) \cup \{uf_i \ | \ u \in k[\sgp_M^{(h)}]_{d-d_i} \ \text{and}\ u\prec m\}).
\]
\end{lemma}









\begin{algorithm}\label{algo:matrixF5}
\SetKwInOut{Input}{Input}
\SetKwInOut{Output}{Output}
\Input{Homogeneous $f_1,\ldots,f_m\in k[\sgp_M^{(h)}]$ of resp. degrees $(d_1,\ldots,d_m)$, a graded monomial ordering $\prec$ on $k[\sgp_M^{(h)}]$,
a maximal degree~$D$}
\Output{a $D$-Gr\"obner basis of $\langle f_1,\ldots,f_m\rangle$ w.r.t. $\prec$}
\lFor{$i=1$ \KwTo $m$}{
$\mathcal{G}_{i}:=\emptyset$
}

 \For{$d=1$ \KwTo $D$}{
$\mathcal{M}_{d,0}:=\emptyset$,  $\widetilde{\mathcal{M}}_{d,0}:=\emptyset$\;

\For{$i=1$ \KwTo $m$}{
\lCase{$d_i>d$\emph{:}}{
$\mathcal{M}_{d,i}:=\widetilde{\mathcal{M}}_{d,i-1}$
}

\lCase{$d_i=d$\emph{:}}{
$\mathcal{M}_{d,i}:=$add new row $f_i$ to $\widetilde{\mathcal{M}}_{d,i-1}$
}

\lCase{$d_i<d$\emph{:}}{
add new row $X^{(s,d-d_i)} f_i$ to $\widetilde{\mathcal{M}}_{d,i-1}$ for all monomials $X^{(s,d-d_i)}\in k[\sgp_M^{(h)}]_{d-d_i}$ that are not in $\langle \LM(\mathcal{G}_{i-1})\rangle$
}

Compute the row echelon form $\widetilde{\mathcal{M}}_{d,i}$ of $\mathcal{M}_{d,i}$\;

Add to $\mathcal{G}_i$ all rows of $\widetilde{\mathcal{M}}_{d,i}$ not top reducible 
by~\(\mathcal{G}_{i}\)\;

}
}
\Return{$\mathcal{G}_m$}
\caption{{\tt sparse-MatrixF5}}
\end{algorithm}
A direct consequence of this lemma is:
\begin{corollary}
Algorithm \ref{algo:matrixF5} is correct.
\end{corollary}
\begin{proof}
With the notations of Algo.~\ref{algo:matrixF5} A direct induction on $d$ and $i$ with Lemma \ref{lem:F5crit} shows that the row span of $\mathcal{M}_{d,i}$ is equal to the row span of the Macaulay matrix in degree $d$ of $(f_1,\ldots,f_i)$. DefProp.~\ref{defprop:dwit} concludes the proof.
\end{proof}

In practice, the choice of the parameter $D$ in Algorithm \ref{algo:matrixF5} is driven by the explicit bounds on the witness degree that we shall derive in Section \ref{sec:complexity}.

\subsection{Sparse-FGLM algorithm}
\label{subsec:sparse-fglm}
The {\tt FGLM} algorithm and its variants might be seen as a tool to change
the representation of a $0$-dimensional ideal. It relies on the notion
of \emph{normal form} relative to an ideal $I$. A normal form relative to $I$ is a $k$-linear map
$\NF:k[S]\rightarrow k[S]$ whose kernel is ${\rm ker}(\NF)=I$. It sends every coset of $I$ to the same representative, allowing effective computations in the ring $k[S]/I$.
One important feature of a sparse Gr\"obner
basis is that it provides a normal form and an algorithm to compute
it by successive reductions of leading monomials.

Let $(p_1,\dots,p_r)$ be the Hilbert basis of a semigroup
$\sgp\subset\Z^n$. Given new indeterminates $H=\{H_1,\dots,H_r\}$, any
monomial in $k[\sgp]$ is the image of a monomial in $k[H]$ via the morphism
$\varphi:k[H_1,\dots,H_r]\rightarrow k[\sgp]$ defined by
$\varphi(H_i)=X^{p_i}$. Given an admissible monomial ordering $\prec_H$ on the ring $k[H_1,\dots,H_r]$,
an ideal $I\subset k[\sgp]$ and a normal form relative to $I$ (given for
instance by a sparse Gr\"obner basis of $I$),
Algorithm~\ref{algo:sparse-fglm} computes a Gr\"obner basis of
$\varphi^{-1}(I)$. Note that $\psi\left(Var(I)\cap
(\overline k^*)^n\right)=Var\left(\varphi^{-1}(I)\right)\cap (\overline
k^*)^r$, where $\psi: \overline k^n\rightarrow \overline k^r$ is the map $\mathbf x\mapsto (\mathbf x^{p_1},\ldots,\mathbf x^{p_r})$. Also, we would like to point out that 
Algorithm~\ref{algo:sparse-fglm} does not depend on the support of the
input sparse system, but only on the ambient semigroup $S_M$.

The main principle of Algorithm~\ref{algo:sparse-fglm} is similar to
the original {\tt FGLM} Algorithm~\cite{FGLM}: we consider the monomials in
$k[H_1,\ldots,H_r]$ in increasing order until we obtain sufficiently
many linear relations between their normal forms. The only difference
is that the computations of the normal forms are performed in
$k[\sgp]$ (using a previously computed sparse Gr\"obner basis) via the
morphism $\varphi$. For solving sparse systems, we choose
the \emph{lexicographical ordering} for $\prec_H$.

\begin{algorithm}\label{algo:sparse-fglm}
  \SetKwInOut{Input}{Input}
  \SetKwInOut{Output}{Output}
  \SetKwComment{Comment}{}{}
   \Input{
     \hspace*{-1mm}%
-a normal form\,$\NF\hspace{-3pt}:k[S]\hspace{-2pt}\rightarrow \hspace{-2pt}k[S]$\,of a $0$-dim ideal~$I$ \\
-a monomial ordering $\prec_H$ on $k[H_1,\dots,H_r]$\\
-a monomial map $\varphi: k[H_1,\dots,H_r]\rightarrow k[\sgp]$\\
  }
  \Output{A Gr\"obner basis in $k[H_1,\dots,H_r]$ w.r.t. $\prec_H$}
  $L:=[1]$\Comment*[l]{\small //list of monomials in \(k[H_1,\dots,H_r]\)}
  $E:=[\,]$\Comment*[l]{\small //staircase for the new ordering $\prec_H$}
  $V:=[\,]$\Comment*[l]{\small //$V=\NF(\varphi(S))$}
  $G:=[\,]$\Comment*[l]{\small //The Gr\"obner basis in $k[H_1,\dots,H_r]$}
  \While{$L\neq [\,]$}{
    $m:=L[1];$ and Remove $m$ from $L$\;
    $v:=\NF(\varphi(m))$\Comment*[r]{\textbf{ (1)}}
    $e:=\#E$ \;
    \eIf{$v\in \Span_{k}\left( V\right)$}{
      $\exists \, (\lambda_{i})\in k^e$ such that $
      v=\sum\limits_{i=1}^{e}\lambda_{i}\cdot V_{i}$\Comment*[r]{\textbf{ (2)}}
      $G:=G\cup \left[
        m-\sum\limits_{i=1}^{s}\lambda _{i}\cdot E_{i}\right]$\;
      Remove from $L$ the elements top-reducible by $G$.
    }{
      $E:=E\cup \lbrack m \rbrack$; \quad
      $V:={\rm RowEchelon}(V\cup \lbrack v \rbrack)$\Comment*[r]{\textbf{ (3)}}
      $L:=\text{Sort}(L\cup \left[ H_{i}\,m\mid
        i=1,\ldots ,r\right] ,\prec_H)$\; 
      Remove from $L$ duplicate elements\;
    }
  }
  Return $G$\;
  \caption{{\tt Sparse-FGLM}}
\end{algorithm}

\begin{theorem}
\label{th:FGLMcorrectness}
 Algorithm {\tt Sparse-FGLM} is correct: it computes the reduced GB of the
ideal $\varphi^{-1}(I)\subset k[H_1,\ldots, H_r]$ with respect to $\prec_H$.
\end{theorem}
\begin{proof}
Let $G=(g_1,\ldots,g_\mu)$ be the output of Algo.~\ref{algo:sparse-fglm}. Set $m_i=\LM(g_i)$. First, we prove that
$G\subset\varphi^{-1}(I)$. Notice that each $g_i$ is of the form
$m_i-q$, where $\varphi(q)=\NF(\varphi(m_i))$. Consequently,
$\NF(\varphi(g_i))=0$ and hence $g_i\in \varphi^{-1}(I)$. Next, let
$h\in k[H]$ be a polynomial such that
$\LM(h)\notin\langle\LM(G)\rangle$. Up to reducing its nonleading
monomials by $G$, we can assume w.l.o.g. that all its monomials do not
belong to $\langle \LM(G)\rangle$. Therefore, the normal forms of the
images by $\varphi$ of all the monomials in the support of $h$ are
linearly independent in $k[S]/I$ (otherwise the linear relation would
have been detected by Algo.~\ref{algo:sparse-fglm}), which means that
$\NF(\varphi(h))\neq 0$ and hence $h\notin \varphi^{-1}(I)$, which
proves that $G$ is a Gr\"obner basis of
$\varphi^{-1}(I)$. The proof that $G$ is reduced is similar.
\end{proof}

\section{Complexity}\label{sec:complexity}
This section is devoted to the complexity of Algorithms \ref{algo:matrixF5} and \ref{algo:sparse-fglm} when the input is a homogeneous regular sequence. In the case of polytopal algebras, the complexity bounds of Theorems \ref{theo:complF5} and \ref{theo:complFGLM} depend mainly on the intrinsic combinatorial properties of the defining polytope.

\smallskip

{\bf Complexity model.} All the complexity bounds count
the number of arithmetic operations $\{+,\times,-,\div\}$ in
$k$; each of them is counted with unit cost. It is not our goal to
take into account operations in the semigroup $\sgp$.

\smallskip

The first goal is to bound $\dwit$ (see DefProp.~\ref{defprop:dwit}) via the Hilbert series of $k[\sgp]/I$. For regular sequences, this Hilbert series can be computed by the following classical formula:
\begin{proposition}\label{prop:HSsuitereg}
Let $\polytope$ be a normal lattice polytope, $f_1,\ldots, f_p\in k[\polytope]$ be a homogeneous regular sequence of homogeneous polynomials of respective degrees $(d_1,\ldots, d_p)$ and $I=\langle f_1,\ldots, f_p\rangle\subset k[\polytope]$. Then 
\[\HS_{k[\polytope]/I}(t)=\HS_{\polytope}(t)\cdot\prod_{i=1}^p (1-t^{d_i}).\]
\end{proposition}
\begin{proof}
See \emph{e.g.} \cite[Exercise 21.17b]{eisenbud1995commutative}.
\end{proof}
The next lemma gives an explicit bound for the witness degree of regular sequences in a normal polytopal algebra:
\begin{lemma}\label{lem:sparsedreg}
  Let $\polytope\subset\R^n$ be a normal lattice polytope and $f_1,\ldots, f_n$ be a homogeneous regular sequence in $k[\polytope]$ of degrees $(d_1,\ldots, d_n)$.
Then any $\left[\reg(k[\polytope])+1+\sum_{j=1}^n(d_j-1)\right]$-sGB  of the ideal $I=\langle f_1,\ldots,f_n\rangle$ is a sGB of $I$. In other words $\dwit\leq \reg(k[\polytope])+1+\sum_{j=1}^n(d_j-1)$.
\end{lemma}

\begin{proof}
By Prop.~\ref{prop:HSsuitereg} and with the notations of Prop.~\ref{prop:HSpolytope}, the
Hilbert series of $k[\polytope]/I$ is equal to
$$\begin{array}{rcl}
\HS_{\polytope}(t)\prod_{i=1}^n (1-t^{d_i})&=&\displaystyle\frac{Q(t)\prod_{i=1}^n
  (1-t^{d_i})}{(1-t)^{n+1}}\\
&=&\displaystyle\frac{Q(1)\prod_{i=1}^n d_i}{1-t}+K(t)
\end{array}
$$ where $K(t)\in\Z[t]$ is a univariate polynomial with
$\deg(K(t))=\reg(k[\polytope])-1+\sum_{i=1}^p(d_i-1)$.  Now, notice that
the Hilbert series of $k[\polytope]/I$ is equal to that of
$k[\polytope]/\LM(I)$. Therefore $\HP_{k[\polytope]/\LM(I)}(d)$ is
constant for $d\geq \deg(K(t))+1$. Since $\ell<\ell'$ implies $\ell\polytope\subset \ell'\polytope$, we obtain
$$\begin{array}{c}\max\{d\in\N\mid \exists X^{(s,d)}\notin \LM(I)\text{ s.t. }\begin{array}[t]{@{}l@{}}
s\in (d\cdot\polytope)\cap\Z^n \text{ and}\\
s\notin ((d-1)\cdot\polytope)\cap\Z^n
\}
  \end{array}\\=\deg(K(t))+1.
\end{array}
$$

Consequently,
minimal generators of $\LM(I)$ and hence minimal homogeneous Gr\"obner bases of
$I$ have degree at most $\deg(K(t))+2=\reg(k[\polytope])+1+\sum_{j=1}^n(d_j-1)$. 
\end{proof}

Now that we have an upper bound for the witness degree, we can estimate the cost of computing a sGB by reducing the Macaulay matrix in degree $\dwit$ (although {\tt sparse-MatrixF5} is a much faster way to compute a sGB in practice, it is not easy to bound precisely its complexity). Note that $\reg(k[\polytope])$ in the following theorem can be deduced from Prop.~\ref{prop:regpolytope}.

\begin{theorem}\label{theo:complF5}
With the same notations as in Lemma~\ref{lem:sparsedreg}, the
complexity of computing a sGB of $\chi_{\polytope\cap\Z^n}(\langle
f_1,\ldots f_n\rangle)\subset k[\sgp_{\polytope\cap\Z^n}]$
by reducing the Macaulay matrix in degree $\dwit$ is bounded above by
$$O\left(n 
\HP_\polytope(\dwit)^\omega\right),$$ where
$\dwit\leq\reg(k[\polytope])+1+\sum_{j=1}^n(d_j-1)$ and $\omega$ is a
feasible exponent for the matrix multiplication ($\omega<2.373$ with \cite{williams2012multiplying}).
\end{theorem}

\begin{proof}
Let $I\subset k[\polytope]$ be the ideal generated by $(
f_1,\ldots,f_n)$.  The number of columns and
rows of the Macaulay matrix in degree $d$ are respectively
$$\begin{array}{rcl}
{\rm nb_{cols}}&=&\HP_\polytope(d),\\
{\rm nb_{rows}}&=&\sum_{i=1}^{n}\HP_\polytope(d-\deg(f_i))\leq n\HP_\polytope(d).
\end{array}$$
Consequently, the row echelon form of such a matrix can be computed within $O(n\HP_\polytope(d)^\omega)$ field operations \cite[Prop. 2.11]{storjohann2000algorithms}. By Proposition~\ref{prop:deshom} and Lemma \ref{lem:sparsedreg}, for $d=\dwit\leq \reg(k[\polytope])+1+\sum_{j=1}^n(d_j-1)$, this provides a sGB of $\chi_{\polytope\cap\Z^n}(I)$.
\end{proof}

We now investigate the complexity of Algorithm~\ref{algo:sparse-fglm}
when $I\subset k[\sgp]$ is a $0$-dim. ideal, and use the same
notations as in Section~\ref{subsec:sparse-fglm}.  Notice that the map
$\varphi$ induces an isomorphism $\psi: k[H]/\varphi^{-1}(I)
\rightarrow k[\sgp]/I$ and therefore Algorithm~\ref{algo:sparse-fglm}
may be seen as a way to change the representation of $k[\sgp]/I$.

\begin{theorem}\label{theo:complFGLM}
Set $\degreeI=\dim_k(k[\sgp]/I)$
 and let $r$ be the cardinality of the Hilbert basis of $\sgp$.
If the input normal form is computed via a reduced sGB of $I\subset k[\sgp]$ (for some monomial ordering), $\sgp$ is a simplicial affine semigroup (see Def.~\ref{def:simplicial}) and $k[\sgp]$ is Cohen-Macaulay, then Algorithm~\ref{algo:sparse-fglm} computes the Gr\"obner basis $G$ with at most $O(r\cdot \degreeI^3)$ operations in $k$.
\end{theorem}
\begin{proof}
Once the $r$ matrices of size $\delta\times\delta$ representing the multiplications by $p_i$ in the canonical monomial basis of $k[S]/I$ are known, Step {\bf{(1)}} in Algorithm~\ref{algo:sparse-fglm} can be achieved in $O(\degreeI^2)$ as in the classical {\tt FGLM} Algorithm~\cite{FGLM}. Steps {\bf{(2)}} and {\bf{(3)}} are done by linear algebra as in~\cite{FGLM}, which leads to a total complexity of $O(r\cdot \degreeI^3)$ since the same analysis holds. It remains to prove that the multiplication matrices can be constructed in $O(r\cdot \degreeI^3)$ operations (this is a consequence of~\cite[Prop.~2.1]{FGLM} in the classical case).  Since $k[\sgp]$ is Cohen-Macaulay and $\sgp$ is simplicial,  we obtain by \cite[Thm.~1.1]{rosales1998cohen} that for any two distinct $p_i,p_j\in \hilbert(\sgp)$ and for any $s\in\sgp$, if $s-p_i$ and $s-p_j$ are in $\sgp$ then $s-p_i-p_j\in\sgp$. With this extra property, the proof of \cite[Prop.~2.1]{FGLM} extends to semigroup algebras. 
\end{proof}

If the input system is a regular sequence of Laurent polynomials, then
$\delta$ can be bounded by the mixed volume of their Newton polytopes
by Kushnirenko-Bernstein's Theorem~\cite{bernshtein1975number}.

\section{Dense, multi-homogeneous and overdetermined systems}\label{sec:applis}
In this section, we specialize Theorems \ref{theo:complF5} and \ref{theo:complFGLM} to several semigroups to obtain new results on the complexity of solving inhomogeneous systems with classical GB algorithms ($\polytope$ is the standard simplex), multi-homogeneous systems ($\polytope$ is a product of simplices) and we state a variant of Fr\"oberg's conjecture for overdetermined sparse systems.

{\bf Inhomogeneous dense systems.}
If $\polytope=\Delta_n$ is the standard simplex in $\R^n$, then
computations of a sparse Gr\"obner basis in the cone over $\Delta_n$
correspond to classical Gr\"obner bases computations using the
so-called ``sugar strategy'' introduced in
\cite{giovini1991one}. The following corollary shows that specializing Theorems \ref{theo:complF5} and \ref{theo:complFGLM} with $\polytope=\Delta_n$ recovers the usual complexity estimates for Gr\"obner bases computations
\begin{corollary}
  Let $f_1,\ldots, f_n$ be a regular sequence of
  polynomials of respective degrees $(d_1,\ldots, d_n)$ in $k[\Delta_n]$. Then the complexity of computing a classical Gr\"obner basis
  of $\langle f_1,\ldots, f_n\rangle$ with respect to a graded
  monomial ordering is bounded by 
$$O\left(n\binom{n+\dwit}{n}^\omega\right),$$
where $\dwit\leq 1+\sum_{i=1}^n(d_i-1)$.
\end{corollary}
In particular, notice that if $f_1,\ldots, f_n$ are inhomogeneous
polynomials in $k[X_1,\ldots, X_n]$, then their homogeneous counterparts in
$k[\Delta_n]$ form a regular sequence if and only if their ``classical
homogenization'' form a regular sequence.

\smallskip

{\bf Multi-homogeneous systems.} Another class of polynomials appearing frequently in applications are \emph{multi-homo\-ge\-neous systems}. A polynomial of multi-degree $(d_1,\ldots, d_\ell)$ w.r.t. a partition of the variables in blocks of sizes $(n_1,\ldots, n_\ell)$ is a polynomial whose Newton polytope is included in $d_1\Delta_{n_1}\times\dots\times d_\ell\Delta_{n_\ell}$. In that case, the associated polytope is a product of simplices, which allows us to state the following complexity theorem:

\begin{theorem}\label{theo:multihom}
Let $f_1,\ldots, f_n$ be a regular sequence of polynomials of multi-degree $(d_1,\ldots, d_\ell)$ w.r.t. a partition of the variables in blocks of sizes $(n_1,\ldots, n_\ell)$ (with $n_1+\dots+n_\ell=n$). Then the combined complexity of Steps (1) to (4) of the solving process in Section~\ref{sec:sGB} is bounded by
$$\begin{array}{rl}
&O\left(n \HP_\polytope(\dwit)^\omega + n\vol(\polytope)^3\right),\\
\text{where}&\polytope=d_1\Delta_{n_1}\times\dots\times d_\ell\Delta_{n_\ell},\\
&\dwit\leq n+2-\max_{i\in\{1,\ldots,\ell\}}(\lceil (n_i+1)/d_i\rceil),\\
&\HP_\polytope(\dwit)=\binom{n_1+\dwit\cdot d_1}{n_1}\cdots\binom{n_\ell+\dwit\cdot d_\ell}{n_\ell},\\
\text{and}&\vol(\polytope)=\binom{n}{n_1,\ldots,n_\ell}\prod_{i=1}^\ell d_i^{n_i}.
\end{array}$$
\end{theorem}

\begin{proof}
Applying Theorems \ref{theo:complF5} and
\ref{theo:complFGLM} with $\polytope$ equal to $d_1\Delta_{n_1}\times\dots\times
d_\ell\Delta_{n_\ell}$ yields the complexity bound in terms of $\dwit$,
$\#\hilbert(S_{\polytope\cap\Z^n})$ and $\delta$.
First, notice that the semigroup generated by $\polytope\cap\Z^n$ is $\N^n$, and hence $\#\hilbert(S_{\polytope\cap\Z^n})=n$.
Next,  $\beta(d_1\Delta_{n_1}\times\dots\times d_\ell\Delta_{n_\ell})$
has an interior lattice point if and only if for all $i$, $\beta
d_i\Delta_{n_i}$ has an interior
lattice point, \emph{i.e.} $\beta d_i > n_i$. The smallest $\beta$ that verifies this condition is
$\max(\lceil (n_1+1)/d_1\rceil,\ldots,\lceil
(n_\ell+1)/d_\ell\rceil)$. By Prop.~\ref{prop:regpolytope}, $\reg(k[\polytope])=n+1-\max(\lceil
(n_1+1)/d_1\rceil,\ldots,\lceil (n_\ell+1)/d_\ell\rceil)$. Since the
$f_1,\ldots, f_n$ have degree $1$ in $k[\polytope]$, we get
$\dwit\leq \reg(k[\polytope])+1$.
Finally, notice that the unnormalized volume of $d
\Delta_q\in\R^q$ is $d^q/q!$. Consequently, the
unnormalized volume of $\polytope$ is $\prod_{i=1}^\ell
d_i^{n_i}/{n_i}!$. Normalizing the volume amounts to multiplying this
value by $n!$, which yields the formula for $\vol(\polytope)$ and equals the multi-homogeneous B\'ezout number. The number of solutions (counted with multiplicity) is classically bounded by this value and hence $\delta\leq \vol(\polytope)$.
\end{proof}

Finally, we state a variant of Fr\"oberg's conjecture
\cite{froberg1985inequality} in the sparse framework, leading to a
notion of ``sparse semi-regularity''. It provides a bound on the witness
degree of generic overdetermined sparse systems: this conjecture can be used to adjust the parameter $D$ of Algorithm~\ref{algo:matrixF5}.

\begin{conjecture}\label{conj:froberg}
Let $\polytope\subset\R^n$ be a normal lattice polytope, $(d_1,\ldots, d_m)\in\N^m$ be a sequence of integers with $m>n$. If $f_1,\ldots, f_m\in \C[\polytope]$ are generic homogeneous polynomials of respective degrees $(d_1,\ldots, d_m)$, then 
$$\HS_{\C[\polytope]/\langle f_1,\ldots, f_m\rangle}(t)=\left[\HS_{\polytope}(t)\prod_{i=1}^m (1-t^{d_i})\right]_+,$$
where $[~]_+$ means truncating the series expansion at its first nonpositive coefficient. 
Systems for which this equality holds are called \emph{semi-regular}. The witness degree of a semi-regular sequence is bounded above by the index of the first zero coefficient in the series expansion of $\HS_{\C[\polytope]/\langle f_1,\ldots, f_m\rangle}(t)$.
\end{conjecture}

\section{Experimental results}\label{sec:expe}
In this section, we estimate the speed-up that one can expect for solving sparse
systems or systems of Laurent polynomials via sparse Gr\"obner bases computations, compared to classical Gr\"obner bases algorithms. 
The same linear algebra routines are used in the compared implementations. Consequently, the speed-up reflects the differences between the characteristics (size, sparseness,\ldots) of the matrices that have to be reduced.

{\bf Workstation.} All experiments have been conducted on a  2.6GHz {\tt IntelCore i7}.

We compare 
{\tt sparse-MatrixF5} (abbreviated {\tt sp-MatrixF5}) with the
implementation of the $F_5$ algorithm in the FGb library. We report
more detailed experimental results on a benchmarks'
webpage\footnote{\url{http://www-polsys.lip6.fr/~jcf/Software/benchssparse.html}}. In all these experiments, the base
field $k$ is the finite field $\GF(65521)$. All tests are done with
overdetermined systems with one rational solution in
$\GF(65521)^n$. The goal is to recover this solution. In that case,
the {\tt FGLM} algorithm is not necessary since the sparse Gr\"obner
basis describes explicitly the image of the solution by a monomial
map. In several settings, we report the speed-up obtained with our
prototype implementation.

{\bf Bilinear systems.} In Table \ref{table:bilinear}, we focus on overdetermined bilinear
systems. For $(n_x,n_y,m)\in\N^3$, we generate a system of $m$ polynomials
with support $\Delta_{n_x}\times\Delta_{n_y}$ uniformly at random in
the set of such systems which have at least one solution in
$\GF(65521)^{n_x+n_y}$.

{\bf Systems of bidegree $(2,1)$.} In Table \ref{table:bidegree}, we
report the performances on overdetermined systems with support $2\Delta_{n_x}\times\Delta_{n_y}$. Note that we obtain important speed-ups when
$n_x<n_y$ (more than 19000 for $(n_x,n_y,m)=(3,10,24)$).

{\bf Fewnomial systems.} In Table \ref{table:fewnomials}, we report
performances on fewnomial systems. The complexity analysis in
Section \ref{sec:complexity} do not apply to this context because the
semigroup algebra in which we compute is not normal. However, the correctness of the algorithms still holds. The
systems are generated as follows: for $(n,t,m)\in\N^3$ we pick $t$
monomials of degree $2$ in $n$ variables uniformly at random and we
generate a system of $m$ polynomials with this support in
$\GF(65521)[X_1,\ldots, X_n]$ with random coefficients such that there
is at least one solution in $\GF(65521)^n$. The computations are done
w.r.t. the semigroup generated by the $t$ monomials. Note that for
some specific instances, the speed-up factor can be as high as 16800.

\begin{table}
  \centering
  \begin{tabular}{|c||c|c||c|}
\hline
$(n_x,n_y,m)$ & {\tt sp-MatrixF5} & {\tt FGb-F5} & Speed-up\\
\hline
\hline
(2,29,40) & 0.12s & 5.2s & 43 \\

(2,39,53) & 0.49s & 36.7s & 74 \\

(2,49,65) & 1.53s & 298.5s & 195\\

(2,59,78) & 4.63s & 852.3s & 184 \\

(6,19,52) & 1.10s & 25.2s & 22 \\

(6,21,56) & 2.13s & 51.5s & 24 \\

(6,27,71) & 7.07s & 236.0s & 33 \\
\hline
  \end{tabular}
  \caption{Overdetermined bilinear systems in $(n_x,n_y)$ variables and $m$ equations}\label{table:bilinear}
\end{table}

\begin{table}
  \centering
  \begin{tabular}{|c||c|c|c|}
\hline
$(n_x,n_y,m)$ & {\tt sp-MatrixF5} & {\tt FGb-F5} &Speed-up \\
\hline
\hline
(1,34,36) & 0.2s & 395.1s & 1975 \\
(1,39,41) & 0.45s & 1641s & 3646 \\
(1,44,46) & 0.75s & 3168.8s & 4225 \\
(2,15,25) & 0.09s & 410.1s & 4556 \\
(2,17,27) & 0.15s & 1894.7s & 12631 \\
(2,19,30) & 0.4s & 5866.1s & 14665 \\
(3,10,24) & 0.15s & 2937.7s & 19584 \\
\hline
(10,4,50) & 23.1s & 1687.3s & 73 \\
(11,5,66) & 155.1s & 6265.8s & 40 \\
(12,6,86) & 872.2s & 27093.3s & 31 \\
\hline
 \end{tabular}
  \caption{Systems in $(n_x,n_y)$ variables of bidegree $(2,1)$ and $m$ equations}\label{table:bidegree}
\end{table}

\begin{table}
  \centering
  \begin{tabular}{|c||c|c||c|}
\hline
$(n,t,m)$ & {\tt sp-MatrixF5} & {\tt FGb-F5} &Speed-up \\
\hline
\hline
(80,240,221)&0.10s & 54.5s&545\\
(80,      240,     223)&           0.08s&               16.3s&203\\        
(150,     450,     434)&           0.24s&              161.2s&671\\
(300,     900,     881)&          4.56s&              11301.0s&2478\\
(120,     240,     233)&           0.01s&               16.8s&16800\\
(40,      160,     128)&            0.21s&               5.93s&28\\
(60,      240,     211)&            0.55s&              29.04s&52\\
\hline
  \end{tabular}
  \caption{Fewnomials systems}\label{table:fewnomials}
\end{table}


\begin{thebibliography}{10}

\bibitem{BFS13}
M.~Bardet, J.-C. Faug{\`e}re, and B.~Salvy.
\newblock On the complexity of the {F5} {G}r{\"o}bner basis algorithm.
\newblock {\em arXiv}, 1312.1655, 2013.

\bibitem{bernshtein1975number}
D.~Bernstein.
\newblock The number of roots of a system of equations.
\newblock {\em Funct. Anal. and its Appli.}, 9(3):183--185, 1975.

\bibitem{brickenstein2010slimgb}
M.~Brickenstein.
\newblock Slimgb: Gr{\"o}bner bases with slim polynomials.
\newblock {\em Revista Matem{\'a}tica Complutense}, 23(2):453--466, 2010.

\bibitem{brodmann2012local}
M.~P. Brodmann and R.~Y. Sharp.
\newblock {\em Local cohomology: an algebraic introduction with geometric
  applications}.
\newblock Cambridge University Press, 1998.

\bibitem{bruns1997normal}
W.~Bruns, J.~Gubeladze, and N.~V. Trung.
\newblock Normal polytopes, triangulations, and {K}oszul algebras.
\newblock {\em J. fur die reine und angewandte Mathematik}, 485:123--160, 1997.

\bibitem{canny1993efficient}
J.~F. Canny and I.~Z. Emiris.
\newblock An efficient algorithm for the sparse mixed resultant.
\newblock In {\em Applied Algebra, Algebraic Algo. and Error-correcting Codes},
  pages 89--104. Springer, 1993.

\bibitem{Canny99asubdivision}
J.~F. Canny and I.~Z. Emiris.
\newblock A subdivision-based algorithm for the sparse resultant.
\newblock {\em J. ACM}, 47:417--451, 1999.

\bibitem{CoxLitSch11}
D.~A. Cox, J.~B. Little, and H.~K. Schenck.
\newblock {\em Toric varieties}.
\newblock AMS, 2011.

\bibitem{EF14}
C.~Eder and J.-C. Faug\`ere.
\newblock A survey on signature-based {G}r\"obner basis computations.
\newblock {\em arXiv}, 1404.1774, 2014.

\bibitem{ehrhart1962polyedres}
E.~Ehrhart.
\newblock Sur les poly{\`e}dres rationnels homoth{\'e}tiques {\`a} n
  dimensions.
\newblock {\em CR Acad. Sci. Paris}, 254:616--618, 1962.

\bibitem{eisenbud1995commutative}
D.~Eisenbud.
\newblock {\em Commutative Algebra: with a view toward algebraic geometry},
  volume 150.
\newblock Springer, 1995.

\bibitem{emiris2005toric}
I.~Z. Emiris.
\newblock Toric resultants and applications to geometric modelling.
\newblock In {\em Solving polynomial equations}, pages 269--300. Springer,
  2005.

\bibitem{emiris2002symbolic}
I.~Z. Emiris and V.~Y. Pan.
\newblock Symbolic and numeric methods for exploiting structure in constructing
  resultant matrices.
\newblock {\em J. of Symbolic Computation}, 33(4):393--413, 2002.

\bibitem{Fau02a}
J.-C. Faug\`ere.
\newblock {A new efficient algorithm for computing Gr\"obner bases without
  reduction to zero (F5)}.
\newblock In {\em ISSAC 2002}, pages 75--83. ACM, 2002.

\bibitem{FGLM}
J.-C. Faug\`ere, P.~Gianni, D.~Lazard, and T.~Mora.
\newblock {Efficient Computation of Zero-dimensional Gr\"obner Bases by Change
  of Ordering}.
\newblock {\em J. of Symb. Computation}, 16(4):329--344, 1993.

\bibitem{FR09}
J.-C. Faug\`ere and S.~Rahmany.
\newblock {Solving systems of polynomial equations with symmetries using
  SAGBI-Gr\"obner bases}.
\newblock In {\em ISSAC '09}, pages 151--158. ACM, 2009.

\bibitem{faugere2011grobner}
J.-C. Faug{\`e}re, M.~Safey El~Din, and P.-J. Spaenlehauer.
\newblock Gr{\"o}bner bases of bihomogeneous ideals generated by polynomials of
  bidegree (1, 1): Algorithms and complexity.
\newblock {\em J. of Symbolic Computation}, 46(4):406--437, 2011.

\bibitem{froberg1985inequality}
R.~Fr{\"o}berg.
\newblock An inequality for {H}ilbert series of graded algebras.
\newblock {\em Mathematica Scandinavica}, 56:117--144, 1985.

\bibitem{fulton1993introduction}
W.~Fulton.
\newblock {\em Introduction to Toric Varieties}.
\newblock Princeton University Press, 1993.

\bibitem{giovini1991one}
A.~Giovini, T.~Mora, G.~Niesi, L.~Robbiano, and C.~Traverso.
\newblock "{O}ne sugar cube, please" or selection strategies in the
  {B}uchberger algorithm.
\newblock In {\em ISSAC '91}, pages 49--54. ACM, 1991.

\bibitem{Hoc72}
M.~Hochster.
\newblock Rings of invariants of tori, {C}ohen-{M}acaulay rings generated by
  monomials, and polytopes.
\newblock {\em The Annals of Mathematics}, 96(2):318--337, 1972.

\bibitem{lazard1983grobner}
D.~Lazard.
\newblock Gr{\"o}bner bases, {G}aussian elimination and resolution of systems
  of algebraic equations.
\newblock In {\em Computer algebra}, pages 146--156. Springer, 1983.

\bibitem{Mac71}
I.~MacDonald.
\newblock Polynomials associated with finite cell-complexes.
\newblock {\em J. London Math. Soc.}, 4:181--192, 1971.

\bibitem{MilStu05}
E.~Miller and B.~Sturmfels.
\newblock {\em Combinatorial commutative algebra}, volume 227.
\newblock Springer Verlag, 2005.

\bibitem{Oda88}
T.~Oda.
\newblock {\em Convex bodies and algebraic geometry}.
\newblock Springer, 1988.

\bibitem{pedersen1995mixed}
P.~Pedersen and B.~Sturmfels.
\newblock Mixed monomial bases.
\newblock In {\em Algorithms in algebraic geometry and applications}, pages
  307--316. Springer, 1995.

\bibitem{rosales1998cohen}
J.~Rosales and P.~A. Garcia-Sanchez.
\newblock On {C}ohen-{M}acaulay and {G}orenstein simplicial affine semigroups.
\newblock {\em Proceedings of the Edinburgh Mathematical Society},
  41(3):517--538, 1998.

\bibitem{stanley1980decompositions}
R.~P. Stanley.
\newblock Decompositions of rational convex polytopes.
\newblock {\em Ann. Discrete Math. v6}, pages 333--342, 1980.

\bibitem{storjohann2000algorithms}
A.~Storjohann.
\newblock Algorithms for matrix canonical forms.
\newblock {\em Ph.D. thesis}, 2000.

\bibitem{Stu91}
B.~Sturmfels.
\newblock Sparse elimination theory.
\newblock In {\em Proc. Comp. Algebraic Geom. and Commut. Algebra}, pages
  377--396. Cambridge Univ. Press, 1991.

\bibitem{Stu96}
B.~Sturmfels.
\newblock {\em Gr{\"o}bner bases and convex polytopes}, volume~8.
\newblock AMS, 1996.

\bibitem{williams2012multiplying}
V.~V. Williams.
\newblock Multiplying matrices faster than {C}oppersmith-{W}inograd.
\newblock In {\em Proc. of STOC'12}, pages 887--898. ACM, 2012.

\end{thebibliography}
\end{document}